\documentclass[pdflatex,sn-aps]{sn-jnl}

\usepackage{amsmath}
\usepackage{amssymb}
\usepackage{amsfonts}
\usepackage{dsfont}
\usepackage{bbm}
\usepackage{physics}

\allowdisplaybreaks

\newcommand{\nn}{\nonumber}

\usepackage{color}

\newcommand{\cv}[1]{(#1)}
\newcommand{\cvb}[1]{[#1]}
\newcommand{\cvc}[1]{\{#1\}}
\newcommand{\cvv}[1]{\vert #1\vert}
\newcommand{\cvV}[1]{\Vert #1\Vert}
\newcommand{\cvr}[1]{\left\langle #1\right\rangle}
\newcommand{\iden}{\mathbbm{1}}
\newcommand{\ptr}[2]{\text{tr}_{#1}\cvc{#2}}
\newcommand{\prob}{\mathbb{P}}
\renewcommand{\op}[1]{\hat{#1}}
\newcommand{\opH}{\op{H}}
\newcommand{\opA}{\op{A}}
\newcommand{\opU}{\op{U}}
\newcommand{\opE}{\op{E}}
\newcommand{\opS}{\hat{\sigma}}
\newcommand{\opP}{\op{P}}
\newcommand{\tra}{\text{tr}}

\jyear{2021}%

\theoremstyle{thmstyleone}%
\newtheorem{theorem}{Theorem}
\newtheorem{proposition}[theorem]{Proposition}%
\newtheorem{lem}[theorem]{Lemma}

\theoremstyle{thmstyletwo}%

\theoremstyle{thmstylethree}%

\begin{document}

\title[Statistics of projective measurement as a noncommutativity witness]{Statistics of projective measurement on a quantum probe as a witness of noncommutativity of  algebra of a probed system}


\author*[1,2]{\fnm{Fattah} \sur{Sakuldee}}\email{fattah.sakuldee@ug.edu.pl}

\author[2]{\fnm{{\L}ukasz} \sur{Cywi\'{n}ski}}\email{lcyw@ifpan.edu.pl}

\affil*[1]{\orgdiv{The International Centre for Theory of Quantum Technologies}, \orgname{University of Gda\'nsk}, \orgaddress{\street{{Jana Ba{\.z}y{\'n}skiego 1A}}, \city{Gda\'nsk}, \postcode{{80-309}}, \country{Poland}}}

\affil[2]{\orgdiv{Institute of Physics}, \orgname{Polish Academy of Sciences}, \orgaddress{\street{al.~Lotnik{\'o}w 32/46}, \city{Warsaw}, \postcode{02-668}, \country{Poland}}}

\abstract{We consider a quantum probe $P$ undergoing pure dephasing due to its interaction with a quantum system $S$. The dynamics of $P$ is then described by a well-defined sub-algebra of operators of $S,$ i.e. the ``accessible'' algebra on $S$ from the point of view of $P.$ We consider sequences of $n$ measurements on $P,$ and investigate the relationship between Kolmogorov consistency of probabilities of obtaining sequences of results with various $n,$ and commutativity of the accessible algebra. For a finite-dimensional $S$ we find conditions under which the Kolmogorov consistency of measurement on $P,$ given that the state of $S$ can be arbitrarily prepared, is equivalent to the commutativity of this algebra. These allow us to describe witnesses of nonclassicality (understood here as noncommutativity) of part of $S$ that affects the probe. For $P$ being a qubit, the witness is particularly simple: observation of breaking of Kolmogorov consistency of sequential measurements on a qubit coupled to $S$ means that the accessible algebra of $S$ is noncommutative.}


\maketitle

\section{Introduction}
In quantum mechanics, one of the fundamental problems is to characterize the quantum and classical features of physical system of interest. One generic characteristic of the quantum system is the presence of noncommutative observables in the operator algebra of the system---the set of all operators forming a vector space \cite{vonNeumann1955,Alicki2008b,Frohlich2015,Aremua2018}. 
The noncommutative structure is the basic ingredient for deriving the Heisenberg uncertainty principle \cite{Griffiths1995,Sakurai1994}, and the violation of Bell's inequality \cite{Bell1964}, as well as classical-quantum discord and related measures \cite{Modi2012}, all of which are used to testify the quantumness of the underlying system. 
In this sense, a verification procedure for noncommutativity of the operator algebra can be considered as a witness for quantumness of the system \cite{Alicki2008a,Alicki2008b,Facchi2012,Facchi2013}. 

An obvious way to probe a system $S$ is to perform $n$ measurements on it (assuming for simplicity that a measurement of the same observable is repeated).
In general, for any given $n,$ we say the probability $\prob_n\cv{m_n,\ldots,m_1}$ of obtaining a sequence of $m_1,\ldots,m_n$ results, satisfies the Kolmogorov consistency (KC) \cite{Feller2008} if for every step $j$
	\begin{align}
	    &\sum_{m_j}\prob_n\cv{m_n,\ldots,m_{j+1},m_j,m_{j-1},\ldots,m_1}\nn\\{}&= \prob_{n-1}\cv{m_n,\ldots,m_{j+1},m_{j-1},\ldots,m_1}. \label{eq:KC}
	\end{align}
It is understood that the probability defined within this context, in general, does not satisfy KC \cite{Breuer2016,Shrapnel2018,Milz2019}, e.g. a joint probability of $n$-step measurement $\prob\cv{m_n,\ldots,m_1}$ defined within quantum framework will generally not satisfy the condition from Eq.~(\ref{eq:KC}). 
In addition, in Ref.~\cite{Strasberg2019} (and similarly in Ref.~\cite{Milz2019}) the local sequential projective measurements (possibly different measurement at different time step), intervening an evolution of open quantum system, and their statistics are considered. 
Eq.~\eqref{eq:KC} is an important property to define many objects, e.g. transition matrices or master equations. 
Hence its consequences on Markovian property of stochastic process become debatable within quantum mechanical setup. 
These lead to two controversial topics in fundamental research related to open quantum system, namely (i) the proper interpretation of Markovian property in quantum probability theory or the role of memory in the dynamics of the open quantum system \cite{Breuer2016,Shrapnel2018,Taranto2019L,Taranto2019A,Pollock2018b} and (ii) the theory of probability within the quantum context \cite{Frohlich2015,Accardi1981,Accardi1982}. 

Recently, it has been suggested that KC can be a well-defined criterion to characterize quantumness and classicality \cite{Strasberg2019,Smirne2019,Milz2019,Milz2020,Sakuldee2018b,Sakuldee2022}. For instance, in Ref.~ \cite{Milz2019}, KC or Eq.~\eqref{eq:KC} is employed as a definition of $n-$classicality for the length $n$ quantum processes $\cvc{\prob_n\cv{m_n,\ldots,m_1}}_{m_n,\ldots,m_1}.$ This observation is motivated by the general observation that, without KC, there can always be nontrivial effects, e.g. interference or quantum back-action, inside the system,  leading to nonclassical feature of the process. 
If the underlying dynamics is completely positive (CP-) divisible and invertible, KC becomes identical to incoherent property of the dynamics---a situation when the experimenter has no ability to access the coherence during the course of the evolution. 
This idea is exemplified in Ref.~\cite{Smirne2019} in both theoretically and experimentally on quantum random walk in optical setup, a quantum walk of a photon on artificial time lattices with attached quantum coin assigned by photon's polarization. For the close system or the dynamical map which is CP-divisible, and for the global measurement on both walker and coin spaces, the fulfilment of KC of the sequential measurement is the absence of quantum coherence. The latter demonstration also exemplifies the implementation of sequential measurements to testify classicality and quantumness of quantum system. 
{We remark here that employing Eq.~\eqref{eq:KC} per se as the definition of the classicality may introduce a fundamental issue concerning the objectivity of this property, i.e. certifying the classicality of the underlying system behind the process requires Eq.~\eqref{eq:KC} to hold for arbitrary $n.$ Checking this is impossible in practice. We will revisit this again in Sect.~\ref{subsec:objectivity}.}

While one can access a quantum system $S$ of interest directly---i.e.~perform a sequence of measurements after preparing it in a desired initial state---a more often encountered situation is when we can coherently control and projectively measure a smaller quantum system---{\it a probe} $P$----that is interacting with a larger and not directly accessible S. For example, qubits undergoing pure dephasing due to interaction with their environment (the system of interest) can be used as probes of dynamics of $S$ when they are subjected to appropriate unitary driving (see e.g.~\cite{Degen_RMP17,Szankowski_JPCM17} and references therein), or when subjected to a sequence of measurements \cite{Facchi_PRA04,Fink_PRL13,Bechtold2016,Zwick_PRAPL16,Sakuldee_classical_PRA20,Sakuldee_operations_PRA20,Do_NJP19,Muller_PLA20}. 

Here we take such a probe-system setup, in which $P$ is a qudit that undergoes dephasing during its evolution in the presence of coupling to a system $S$ of finite dimension. 
We investigate the connection between the Kolmogorov consistency of probabilities of obtaining sequences of results of measurements on $P$ and classicality (defined here as commutativity) of a sub-algebra of operators of $S$. An interesting feature of the $P$-$S$ setup discussed here is that such a sub-algebra of operators of $S$, the properties of which are in principle imprinted on dynamics of $P$, appears in a natural way---it is picked out by the form of $P$-$S$ coupling and the Hamiltonian of $S$. By considering such a setup we are thus asking a question ``is a system $S$ classical when it is observed from the point of view of a probe $P$?''.

The two main results of the paper are the following. 
(1) When all the observables $\opE_m$ on the system $S$ induced by projectives measurement on the probe $P$ for outcomes $m,$  are nondegenerate, fulfilment of all the Kolmogorov consistency conditions for sequences of measurements on $P$, for all possible states of $S$, is equivalent to the accessible algebra being commutative. (2) In the often-encountered in experiments case of probe being a qubit that couples through its $\hat{\sigma}_z$ operator to the system, fulfilment of KC conditions for sequences of repeated measurements along two orthogonal axes (e.g.~repeated measurement of $\hat{\sigma}_x$ and $\hat{\sigma}_y$ observables), for all states of $S$, is equivalent to commutativity of the accessible algebra of $S$. From these results we propose a witness of quantumness of $S$ from the point of view of $P$ (understood as noncommutativity of the accessible algebra): breaking of KC for measurements on a qubit, or for measurements on a $d$-dimensional probe with $d>2$ when  $\opE_m$ are nondegenerate, proves that $[\hat{H}_i,\hat{H}_j] \! \neq \! 0$ for some $i$ and $j$.  

The paper is organized in the following way. Mathematical framework and terminologies for classicality, measurement protocol on a qudit undergoing pure dephasing and related topics are given in Sect. \ref{sec:prelim}. 
We discuss there the commutativity as a notion of classicality of some region in the operator space and KC for the underlying process concerning such a region. 
Then, in Section \ref{sec:KC_com} we demonstrate the relation between these two for our $P$-$S$ model with $P$ undergoing pure dephasing.  In Sect.~\ref{sec:QW} we describe the ways in which the noncomutativity (nonclassicality) of the accessible algebra of $S$ can be witnessed by observing the breaking of KC conditions for measurements on $P$. Discussion of various issues related to witnessing nonclassicality of $S$ (e.g.~creation of $P$-$S$ entanglement, possibility of replacing the influence of $S$ on $P$ by classical noise, Leggett-Garg inequalities) is given in Sect.~\ref{sec:discussion}. 

\section{Accessible algebra and induced measurement sequences}\label{sec:prelim}
In this section we establish a mathematical platform for our problem. The basic definitions and formalism are given in Sect.~\ref{subsec:quantum_system}. 
The formulation of an open quantum dynamics and induced measurement operations on the system is introduced in Sect.~\ref{subsec:measurement_protocol}.
The measurement statistics for general sequential measurement is briefly given in Sect.~\ref{subsec:measurement_induced_observable} where the role of induced measurement is also exemplified for the qubit. Lastly, conditions for existence of a fixed point of measurement maps that we consider here, which will be used in subsequent Section, are discussed in Sect.~\ref{subsec:RU}.

\subsection{Quantum systems}\label{subsec:quantum_system}
In this work, we consider a linear system described by a Hilbert space $\mathcal{H}$ of finite dimension $d$, $\mathcal{H}\simeq\mathbb{C}^{d}$ \cite{Thirring1981}. We call a \emph{quantum system} a doublet $\cv{\mathcal{A},\rho}$ 
consisting of an operator algebra $\mathcal{A}=\mathcal{B}\cv{\mathcal{H}}\simeq\mathcal{M}_{d}\cv{\mathbb{C}}$ or a Banach algebra $\mathcal{B}\cv{\mathcal{H}}$ of bounded operators on such Hilbert space, equipped with a Hermitian conjugation $\cdot^\dagger$ and a norm $\cvV{\cdot},$ where $\mathcal{M}_{d}\cv{\mathbb{C}}$ denotes a ($d-$square) matrix representation (matrix algebra); and $\rho$ is a density matrix $\rho>0,$ $\tr\cv{\rho}=1,$ $\rho=\rho^\dagger.$   
Let $\mathcal{K}$ denote a map on the algebra and assume that it is completely positive and trace nonincreasing. Equivalently, for a given density matrix $\rho,$ we can represent such a map in the operator sum form $\mathcal{K}\cvb{\rho}=\sum_m \op{K}_m\rho\op{K}_m^\dagger,$ where its dual map on operator space takes the form $\mathcal{K}^\dagger\cvb{\opA}=\sum_m\op{K}_m^\dagger\opA\op{K}_m$ derived from the Schr\"odinger-Heisenberg correspondence $\tr\cv{\rho~\!\mathcal{K}^\dagger\cvb{\opA}}=\tr\cv{\mathcal{K}\cvb{\rho}\opA}.$ 
Note that, in this language, the output operator for a given 
$m,$ $\mathcal{K}_{m}\cvb{\rho}=\op{K}_m\rho\op{K}_m^\dagger$ needs not to be a density matrix but rather an observable (a non-negative operator) counterpart of a probability associated with $m$ in operator algebra.

The operator algebra $\mathcal{A}$ is said to be \emph{classical with respect to a state} $\rho,$ if $\tr\cv{\rho\cvb{\opA,\op{B}}}=0$ for all $\opA,\op{B}\in\mathcal{A}.$ If it is classical with respect to all possible states we say $\mathcal{A}$ is \emph{classical}. It is clear that if the algebra $\mathcal{A}$ is commutative it will be classical \cite{Facchi2012}, and vice versa. If $\mathcal{A}$ is classical (commutative), from the Gel'fand theorem \cite{Thirring1981,Alicki2008b}, such algebra will be isomorphic to a $C^*-$algebra of continuous functions on a certain compact subset (e.g. a set of characters of the algebra), or in other words, all operators inside $\mathcal{A}$ can be represented by a function over a field. 
We can define a sub-algebra $\mathcal{A}_{\mathcal{L}}$ from a given set of operators $\mathcal{L}=\cvc{\iden}\cup\cvc{\op{L}_i}_i,$ also known as a sub-algebra $\mathcal{A}_{\mathcal{L}}$ generated by $\mathcal{L},$ by taking all possible complex linear combinations of all possible  products of finite numbers of elements in the set $\mathcal{L}$ together with all limit points of any convergent sequences, and the set $\mathcal{A}_{\mathcal{L}}$ will become a matrix-sub-algebra, namely $\mathcal{A}_{\mathcal{L}}\subset\mathcal{A}$ \cite{Frohlich2015}. If $\mathcal{A}$ is classical, so is $\mathcal{A}_{\mathcal{L}}$ for all possible $\mathcal{L},$ but the converse needs not to be true. We remark here that throughout this paper the terms classical and commutative can be used interchangeably.

\subsection{Pure dephasing interaction and induced measurement map}\label{subsec:measurement_protocol}
We are interested in the case when the system ($S$) of interest $\cv{\mathcal{A},\rho}$ is not directly accessible to us---we can only control and measure another system---the probe ($P$) that is coupled to it. Such a setting arises in a natural way e.g.~when we consider a qudit probe $P$ coupled to its environment---in this situation the lack (or very limited nature) of control over the 
system $S$ is basically a definition of environment. For the composite system  $\mathcal{A}_P\otimes\mathcal{A}_S\simeq\mathcal{B}\cv{\mathcal{H}_P}\otimes\mathcal{B}\cv{\mathcal{H}_S}$ (where $\mathcal{A}_P\simeq\mathcal{B}\cv{\mathcal{H}_P}$ is an operator algebra of our controllable system of dimension $d$), we consider a unitary evolution given by 
	\begin{equation}
		\opU=\sum_{i=1}^d\ketbra{i}\otimes\opU_i \,\, \label{eq:pure-dephase-U}.
	\end{equation}
generated by qudit pure dephasing type Hamiltonian $\opH=\sum_{i=1}^d\ketbra{i}\otimes\opH_i,$ i.e. $\opU_i=e^{-it\opH_i}$ for time duration $t,$ where $\opH_i=\opH_i^\dagger$ are Hamiltonians in the system algebra $\mathcal{A}$ and $\cvc{\ketbra{i}}_i$ is a set of orthonormal projection on $\mathcal{H}_P.$ 
Such an evolution occurs when the Hamiltonian of the composite system is given by $\hat{H}_P+\hat{H}_S+\hat{V}_{P}\otimes\hat{V}_S$, where $\hat{H}_{S(P)}$ are the Hamiltonians of the system of interest (the qudit that is under our control), and $\hat{V}_{P}\otimes\hat{V}_S$ is the coupling between $P$  and $S$ . When $\hat{H}_{P}$ and $\hat{V}_P$ commute, their common eigenstates $\ket{i}$ are pointer states unperturbed 
by interaction with the system \cite{Zurek2003}, and $\opH_i = \epsilon_i + \hat{H}_S + v_i \hat{V}_S$, with $\epsilon_i$ ($v_i$) being eigenvalues of $\hat{H}_P$ ($\hat{V}_P$). 

Within this setup we define an \emph{accessible sub-algebra} $\mathcal{A}_{acc}=\mathcal{A}_{\mathcal{L}}$ with $\mathcal{L}=\cvc{\iden}\cup\cvc{\opH_i}_{i=1}^d.$ 
The term ``accessible'' is employed in the sense that the information content of the system  accessible by the pure dephasing Hamiltonian $\opH$ cannot exceed the algebra $\mathcal{A}_{acc};$ however, this is a slight abuse of terminology, since the whole $\mathcal{A}_{acc}$ may not be accessed arbitrarily by the probe. The accessible algebra is the central object defining the connection between the features of the system, and the dynamics of the probe. Since the probe is sensitive to the system only through evolution operator Eq.~\eqref{eq:pure-dephase-U}, elements of algebra of $S$ that are beyond $\mathcal{A}_{acc}$ have no influence on behaviour of $P.$  
We also assume  that the initial state $\rho$ of the system is arbitrary, so the characterization of the quantum pair $\cv{\mathcal{A},\rho}$ is reduced to the study of $\mathcal{A}_{acc}$ solely.

Let us now consider the prepare-evolve-measure protocol on the qudit and its reduced map on the system. Given a prepared state from an orthonormal basis $\cvc{\ket{\gamma}}$ of the Hilbert space $\mathcal{H}_P,$ and a measured state from (possibly) another orthonormal basis $\cvc{\ket{m}}$ (they are posterior states of the measurement of a quantity $\sum_m m\ketbra{m}{m}$ on the probe $P,$) and assuming  the product input state, $\ketbra{\gamma}\otimes\rho$, we arrive at an induced map on the system 
	\begin{equation}
		\op{K}_{m,\gamma}:=\ptr{P}{\cv{\ketbra{\gamma}{m}\otimes\iden}\op{U}}=\sum_i\overline{\gamma}(i)m(i)\op{U}_i \,\, \label{eq:Km}
	\end{equation}
where $\gamma(i)=\braket{\gamma}{i},$ $m(i)=\braket{m}{i}$ and $\overline{a}$ is a complex conjugate of a complex number $a.$ Note that $\ket{\gamma}$ or $\ket{m}$ need not be orthogonal to each other. The operations and the effects, completely positive maps and observables corresponding to the prepare-measure index pair $\cv{\gamma,m},$ can then be written as
		\begin{align}
			\mathcal{K}_{m,\gamma}\cvb{\rho} &= \op{K}_{m,\gamma}\rho\op{K}_{m,\gamma}^\dagger= \sum_{i,i'}\left \cvb{\overline{\gamma}(i)m(i)\overline{m}(i')\gamma(i')\right }\op{U}_i\rho\op{U}_{i'}^\dagger,\label{eq:K_map_gen}\\
			\op{E}_{m,\gamma} &= \sum_{i,i'}\left \cvb{\overline{\gamma}(i)m(i)\overline{m}(i')\gamma(i')\right }\op{U}_{i'}^\dagger\op{U}_i,\label{eq:Effect_gen}
		\end{align}
	for a given state $\rho.$ 
We term the prepare-evolve-measure protocol with the condition above as that of \emph{complete projections} type. 
For simplicity, we consider a fixed prepared state $\ket{\gamma};$ the effects $\cvc{\op{E}_{m,\gamma}}_m$ form then a positive operator-valued measures (POVMs) \cite{Nielsen2010}, where we also have
	\begin{equation}
		\sum_m\op{E}_{m,\gamma}=\iden,  \label{eq:sumE}
	\end{equation}
where we will drop the parameter $\gamma.$
In other words, the maps $\mathcal{K}_{m,\gamma}$ correspond to weak measurements on the system \cite{Wiseman2009}. 

As an example, we illustrate the above using measurements on a qubit. 
In such a case we write $\opH=\ketbra{\uparrow}\otimes\opH_\uparrow+\ketbra{\downarrow}\otimes\opH_\downarrow$ and $\opU=\ketbra{\uparrow}\otimes\opU_\uparrow+\ketbra{\downarrow}\otimes\opU_\downarrow$ where $\ket{\uparrow}$ and $\ket{\downarrow}$ are eigenstates of Pauli operator $\opS_z,$ and our prepared state is set as $\ket \gamma=\ket {+x}=\cv{\ket \uparrow + \ket \downarrow}/\sqrt{2}.$ We set also the measurement basis $\cvc{\ket{m}}$ from the eigenstates of $\opS_x=\ketbra{\uparrow}{\downarrow}+\ketbra{\downarrow}{\uparrow},$ denoted by $\ket{\pm x}$ with eigenvalues $\pm 1.$ A measurement in $Y$ axis from the eigenstates of $\opS_y=i\ketbra{\uparrow}{\downarrow}-i\ketbra{\downarrow}{\uparrow}$ denoted by $\ket{\pm y}$ with eigenvalues $\pm 1,$ can also be applicable in the same setting. 
In particular, we write 
    \begin{equation}
        \op{K}^X_\pm=\dfrac{1}{2}\cv{\op{U}_\uparrow \pm \op{U}_\downarrow},~~\opE^X_\pm=\dfrac{\iden}{2}\pm\dfrac{1}{4}\cv{\op{U}_\uparrow^\dagger\op{U}_\downarrow + \op{U}_\downarrow^\dagger\op{U}_\uparrow},
    \end{equation}
for $X-$ measurement and 
    \begin{equation}
        \op{K}^Y_\pm=\dfrac{1}{2}\cv{\op{U}_\uparrow \pm i\op{U}_\downarrow},~~\opE^Y_\pm=\dfrac{\iden}{2}\pm\dfrac{i}{4}\cv{\op{U}_\uparrow^\dagger\op{U}_\downarrow - \op{U}_\downarrow^\dagger\op{U}_\uparrow},
    \end{equation}
for $Y-$ measurement. 

\subsection{Measurement statistics and induced system observable}\label{subsec:measurement_induced_observable}
Now we consider the sequences of measurements on the probe  obtained by repeatedly applying the measurement map $\mathcal{K}_{m}.$ Given a sequence of measurement outcomes $\cv{m_n,\ldots,m_1} ,$ the probability for the measurement sequence can be obtained quantum mechanically \cite{Frohlich2015} as 
	\begin{align}
		\prob_n\cv{m_n,\ldots,m_1}&=\tr\left\cvb{\rho\!~\mathcal{K}^\dagger_{m_1}\circ\cdots\circ\mathcal{K}^\dagger_{m_n}\cvb{\iden}\right}\label{eq:prob_def}\\
		&=\tr\cv{\rho~\!\op{Q}_n\cv{m_n,\ldots,m_1}}\label{eq:prob_in_Q}
	\end{align}
where \[\op{Q}_n\cv{m_n,\ldots,m_1}=\op{R}_n^\dagger\cv{m_n,\ldots,m_1}\op{R}_n\cv{m_n,\ldots,m_1}\] and $\op{R}_n\cv{m_n,\ldots,m_1}=\op{K}_{m_n}\cdots\op{K}_{m_1}.$ 
The operator $\op{R}$ can be interpreted as a quantum history as it is an analogue of classical history in the standard stochastic analysis \cite{Frohlich2015,Smirne2019}, and indeed, the probability given above follows simply from the Born rule associated with the observable $\op{Q}_n.$ 
Let us remark that the measurement probability of single outcome in the sequence of $n>1$ cannot be assigned, except when the 
KC is satisfied by all the $\prob_n\cv{m_n,\ldots,m_1}$, e.g. one cannot construct a probability $\prob\cv{m_j}$ for $1\leq j\leq n$ out of the family $\cvc{\prob\cv{m_n,\ldots,m_1}}$ in general.

In the above setup, the measurements are taken on the probe qudit by a sequence of one dimensional projections with respect to the set of pure states $\cvc{\ket{m_k}}_{k=1}^n,$ and they thus induce a set of weak measurements on the system operator algebra $\mathcal{A}_{acc}.$ 
Let us remark that $\op{Q}_1(m_n)=\op{E}_{m_n}$ for the step n, and for the other steps the observable is a mapping of previous history by an induced measurement operation, i.e. $\op{Q}_{k}\cv{m_k,\ldots,m_1}=\mathcal{K}^\dagger_{m_k}\cvb{\op{Q}_{k-1}\cv{m_{k-1},\ldots,m_1}}.$ 
This is an example of an update of quantum history (trajectory) in the sense of Ref. \cite{Sakuldee2018b}. 
Note that within this formalism there will always be time arguments implicitly present in the definitions of evolution maps $\mathcal{K}_{m_k}$, and here the time intervals between preparation and measurement in each step of measurement protocol are identical. 

In the case of the probe being a qubit, it is natural to choose the measurements as projections along qubit's $X-$ or $Y-$ axes, with measurement axes being possibly distinct in each step.
The generic form of probability induced by the qubit can be written as $\prob_n\cv{m_n^{\alpha_n},\ldots,m_1^{\alpha_1}}$ where the superscript $\alpha_k \!= \! X,Y$ denotes the axis of measurement at the time step $k.$
We write $\prob^{X(Y)}_n\cv{m_n,\ldots,m_1}$ if all measurements are taken along the same axis $X(Y).$ 

\subsection{Nonselective measurement}\label{subsec:RU}
Let us discuss now the case of a nonselective measurement \cite{Wiseman2009,vonNeumann1955}, i.e.~one that is performed, but its results are ignored. 
This type of maps arises from the summation over intermediate outcomes in KC condition. It can be written as $\mathcal{K}=\sum_m\mathcal{K}_m,$ and this map will behave no differently from an identity map upon the observable of concern $Q,$ when KC is fulfilled. 
In our complete projections protocol, such a nonselective measurement map is given by
	\begin{align}
		\mathcal{K}\cvb{\rho} = \sum_{i=1}^d \cvv{\gamma\cv{i}}^2\opU_i\rho\opU_i^\dagger, \label{eq:non-select-Kd}\\
		\mathcal{K}^\dagger\cvb{\opA} = \sum_{i=1}^d \cvv{\gamma\cv{i}}^2\opU^\dagger_i\opA\opU_i,\label{eq:non-select-K}
	\end{align}
with $\sum_i\cvv{\gamma\cv{i}}^2=1$ set by prepared state. Both operations belong to the class of quantum operation called random unitary channels, which are completely positive, trace preserving, and unital maps (the maps that preserve the identity.) The spectral structure and analysis, and fixed point properties for the maps in this class were examined in detail in the literature, see e.g. \cite{Bialonczyk2018,Mendl2009,Zyczkowski1994,Audenaert2008}. 
When KC given in Eq.~\eqref{eq:KC} holds, the action of such a nonselective measurement is the same as that of an identity channel. Hence it is interesting to consider operators that are invariant under $\mathcal{K}^{\dagger}$.

In general, it is proven that, under some assumptions, for a given quantum operation (completely positive and trace preserving map), the set of all fixed points of the operation is identical to the set of operators commuting with all the Kraus operators associated with the operation (see Refs.~\cite{Arias2002,Heinosaari2010} for criteria needed for the equivalence). 
One case where such an equivalence holds, is a unital map on finite dimension (this includes the case of a random unitary map), which the structure of our map (see, for instance, Theorem 3.6 (b) in Ref.~\cite{Arias2002}).
In other words, for quantum operation $\Phi_{\mathcal{G}}\cvb{\opA}=\sum_j\op{G}_j\opA\op{G}_j^\dagger$ and the set of its corresponding Kraus operators $\mathcal{G}=\cvc{\op{G}_j,\op{G}_j^\dagger}_j,$ the set of all fixed point $\mathcal{B}\cv{\mathcal{H}}^{\Phi_{\mathcal{G}}}:=\cvc{\opA\in\mathcal{B}\cv{\mathcal{H}} : \Phi_{\mathcal{G}}\cvb{\opA}=\opA}$ is identical to the commutant $\mathcal{G}':=\cvc{\opA\in\mathcal{B}\cv{\mathcal{H}} : \cvb{\opA,\op{G}_j}=0 \text{~for all~}j}.$ It can be also said that the operator $\opA$ is invariant under $\Phi_\mathcal{G}$ if and only if it is invariant under $\op{G}_j\cdot\op{G}_j^\dagger$ for all $j.$
For our case we have then that for any $i=1,\ldots,d,$ an operator $\opA$ is invariant under operation $\mathcal{K}^\dagger$ if and only if $\cvb{\opH_i,\opA}=0$ and $\rho$ is a steady state of $\mathcal{K}$ if and only if $\cvb{\opH_i,\rho}=0.$

\section{Kolmogorov consistency and commutativity of the accessible algebra} \label{sec:KC_com}

Let us start by noting that when the accessible sub-algebra $\mathcal{A}_{acc}$ is commutative, the probability defined by Eq.~\eqref{eq:prob_in_Q}  satisfies KC, Eq.~\eqref{eq:KC}. In order to see this, it is helpful to rewrite Eq.~\eqref{eq:prob_in_Q} as 
\begin{align}
	\prob_n\cv{m_n,\ldots,m_1}&=\tra\cvb{\rho \hat{K}^\dagger_{m_1} \ldots \hat{K}^\dagger_{m_n} \hat{K}_{m_n} \ldots \hat{K}_{m_1} }  = \tra\cvb{\rho \hat{K}^\dagger_{m_1} \ldots  \hat{E}_{m_n} \ldots \hat{K}_{m_1} } \label{eq:P_K}
	\end{align}
Commutativity of $\mathcal{A}_{acc}$ implies that $[\hat{K}_m,\hat{E}_{m'}] \! = \! 0$ for every $m$ and $m'$. Using the above form for $\prob_{n}(m_n,\ldots,m_1)$ on the left-hand-side of Eq.~\eqref{eq:KC} we can then commute $\hat{E}_{m_n}$ through all the $\hat{K}^{\dagger}$ operators, repeat this procedure with $\hat{E}_{m_{n-1}} = \hat{K}^\dagger_{m_{n-1}}\hat{K}_{m_{n-1}}$, and continue until we arrive at the form of the joint probability that hold in the commutative case:
\begin{equation}
\prob_n\cv{m_n,\ldots,m_1} = \tra\cvb{\rho \hat{E}_{m_n} \ldots \hat{E}_{m_1}  } \,\, . \label{eq:P_E}
\end{equation}
Using then Eq.~\eqref{eq:sumE}, i.e. that $\cvc{\hat{E}_{m_j}}_{m_j}$ form POVMs for all $j,$ we see that this $\prob_n$ is equal to an appropriate $\prob_{n-1}$ when summed over any $m_j$, i.e.~Eq.~\eqref{eq:KC} holds for every $j$.

{
 A well-known example of a commutative (classical) system considered above is the situation when the system can be described by functions. To see this let us consider a qubit probe whose dynamics 
 are generated by a time-dependent Hamiltonian $\hat{H}\cv{t} = \xi(t)\hat{\sigma}_z,$ where $\xi(t)$ describes the influence of the system onto the qubit \footnote{The Hamiltonian in this example is time dependent, however, thanks to the commutative nature of the system, the analysis is the same as the one from Eq.~\eqref{eq:pure-dephase-U} with mutually commuting conditioned Hamiltonians.}. Here we have that 
    \begin{equation*}
    \hat{K}_{m_k} = \dfrac{1}{2}\big(e^{-i\alpha_k} + m_k e^{i\alpha_k}\big) = \left\{\begin{array}{lr}
        \cos\alpha_k, & m_k=1,  \\
        -i\sin\alpha_k, & m_k=-1,
    \end{array} \right. 
    \end{equation*}
where the qubit state is prepared in $\ket{+x}$ and measured in $\ket{m_k x}$ with $m_k=\pm1,$ $\alpha_k=\displaystyle\int_{t_k}^{t_{k+1}}\xi(t)\mathrm{d}t,$ and $t_{k+1}-t_k$ is the duration of the evolution that occurs between initialization at time $t_k$ and measurement at time $t_{k+1}.$ By commutativity of functions, it follows that
    \begin{equation*}
    \hat{E}_{m_k} = \left\{\begin{array}{lr}
        \cos^2\alpha_k, & m_k=1,  \\
        \sin^2\alpha_k, & m_k=-1,
    \end{array} \right. 
    \end{equation*}
and hence $\sum_{m_k}\hat{E}_{m_k}=1.$ With this at hand, via Eq.~\eqref{eq:P_E}, one can see that the KC can be fulfilled regardless of the time dependence of function $\xi\cv{t}.$ This fact holds when $\xi\cv{t}$ is a general classical stochastic process (with no assumptions made about its correlation time or Markovian vs non-Markovian nature), and calculation of expectation values involves an additional averaging over realizations of this process.}

{From the arguments above,} it follows that if the probability over measurement sequence fails to satisfy KC, the underlying sub-algebra cannot be commutative.
It is obviously  interesting to check under which conditions KC condition implies the commutativity of the accessible algebra.
To do so we consider a general case of arbitrary, either commutative or noncommutative, $\mathcal{A}_{acc}$ but with the probability $\prob_n\cv{m_n,\ldots,m_1}$ of measurement sequences given by Eq.~\eqref{eq:prob_in_Q} satisfying KC condition.

\subsection{Sufficient conditions for Kolmogorov consistency to imply classicality of the accessible algebra}\label{subsec:equi_random_unitary}
We observe that,  for arbitrary $\rho,$ Eq.~\eqref{eq:KC} with $j\! =\! 1$ implies
		\begin{equation}
			\mathcal{K}^\dagger\cvb{\op{Q}_{n-1}\cv{m_n,\ldots,m_2}} = \op{Q}_{n-1}\cv{m_n,\ldots,m_2} \label{eq:necessary_Kol}
		\end{equation}
for all $n.$ 
{This condition that is a necessary one for KC will be useful in the following analysis of sufficient conditions.} 
	\begin{lem}\label{props:single_axis}
		Given $\prob_n\cv{m_n,\ldots,m_1}$, a measurement probability associated with a complete projections protocol, with identical measurement basis for every time step, KC for $\prob_n,$ for arbitrary state $\rho,$ is equivalent to \[\cvb{\opH_i,\op{E}_m}=0,\]
		for all $i$ and $m.$
	\end{lem}
	\begin{proof}
		Assume that $\cvb{\opH_i,\op{E}_m}=0,$ for all $i$ and $m,$ which implies $\cvb{\op{K}_{m'},\op{E}_{m}}=0,$ for all $m$ and $m',$ KC follows then according to reasoning giving Eqs.~\eqref{eq:P_K} and ~\eqref{eq:P_E}.
For the converse, it is enough to use Eq.~\eqref{eq:necessary_Kol}  for $n=2$, which reads
    \begin{equation}
        \sum_{m'}\op{K}^\dagger_{m'}\op{E}_{m}\op{K}_{m'} = \op{E}_{m} \label{eq:necessary_Kol_reduced},
    \end{equation}
for all $m.$ From the discussion in Sect.~\ref{subsec:RU} we know that it is equivalent to $\cvb{\op{K}_{m'},\op{E}_{m}}=0$ for all $m$ and $m'.$ 
For the consistency of the article, let us consider the following. Define an inner product $\cvr{\op{A},\op{B}}=\sqrt{\tr\cv{\op{A}^\dagger\op{B}}}$ 
and set $\cvV{A}_2^2:=\cvr{\op{A},\op{A}}.$ By triangle and Cauchy-Schwarz inequalities, if $\mathcal{K}\cvb{\op{E}_{m}}=\op{E}_{m}$ we have
\begin{align*}
	\cvV{\op{E}_{m}}_2^2 &= \Big\vert\cvr{\op{E}_{m},\mathcal{K}\cvb{\op{E}_{m}}}\Big\vert\leqslant \sum_{i=1}^d \cvv{\gamma\cv{i}}^2\Big\vert\cvr{\op{E}_{m},\mathcal{U}_i\cvb{\op{E}_{m}}}\Big\vert\leqslant \cvV{\op{E}_{m}}_2^2.
\end{align*}
where we recall Eq.~(\ref{eq:non-select-Kd}). Since the equality holds, then $\op{E}_{m}$ and $\mathcal{U}_i\cvb{\op{E}_{m}}$ are linearly dependent, e.g. $\mathcal{U}_i\cvb{\op{E}_{m}}=\lambda\op{E}_{m}$ for some positive number $\lambda.$ By trace preserving property we have $\lambda=1$ and then
$\sum_i \tilde{U}_i^\dagger \op{E}_{m} \tilde{U}_i = \op{E}_m$, which is equivalent to $[\op{E}_{m},\tilde{U}_i]\! =\! 0$ for all $m$ and $i$, where $\tilde{U}_i \! =\! \gamma\cv{i} \op{U}_i$ form a set of Kraus operators. From this, $[\op{H}_i,\op{E}_{m}]\! = \! 0$ for all $i$ and $m$ follows immediately.
	\end{proof}

Let us note that from the structure of the above proof it follows that fulfilling KC for $n=2$ and $j=1$ {\it and for any $\rho$} implies that all the consistency conditions (for any $n$ and $j$ are fulfilled). This strong result follows from our focus on statements that hold for an arbitrary state $\rho$ of the system.

Let us now focus on the probability of obtaining measurement sequences in Lemma \ref{props:single_axis} and demonstrate the conditions for equivalence. \begin{proposition}\label{props:nondegenerate}
	For $\opE_m$ that are nondegenerate for all $m,$ the probability $\prob_n\cv{m_n,\ldots,m_1}$ given by complete projections protocol, with identical measurement basis for every time step, satisfies KC for arbitrary state $\rho$ if, and only if, $\mathcal{A}_{acc}$ is commutative.
\end{proposition}
\begin{proof}
	As discussed, it suffices to show that KC implies commutative structure.
First we remark that, in general, by Jacobi identity the conditions $\cvb{\opE_m,\opH_j}=\cvb{\opE_m,\opH_i}=0$ lead to $\cvb{\opE_m,\cvb{\opH_i,\opH_j}}=0.$ Now we will show that $\cvb{\opH_i,\opH_j}$ is a zero matrix. Let $\opE_m=\displaystyle\sum_{\ell=1}^de_{\ell}(m)\opP_{\ell}(m)$ be a spectral decomposition of the observable $\opE_m$ for any $m.$ Since the Hamiltonian $\opH_i$ is Hermitian, then it will share common eigenvectors with $\opE_m,$ and so does $\opH_j.$ From the assumption that $\opE_m$ is nondegenerate, i.e. $e_{\ell}\cv{m}$-eigensubspace is one dimensional, then we have 
    \begin{equation}
        \opP_{\ell}\cvb{\opH_i,\opH_j}\opP_{\ell'}=\cvb{h_{\ell}\cv{i}h_{\ell}\cv{j}-h_{\ell}\cv{j}h_{\ell}\cv{i}}\delta_{\ell\ell'}\opP_{\ell}=0, \label{eq:invariant_condition_nondegenerate}
    \end{equation}
where we write $h_{\ell}\cv{i}=\tra\cv{\opP_{\ell}\opH_i},$ $h_{\ell}\cv{j}=\tra\cv{\opP_{\ell}\opH_j}$ and $\delta_{\ell\ell'}$ is a Kronecker symbol. This is not the case when the observable $\opE_m$ is degenerate since there could be at least one eigensubspace with dimension higher than one. In such scenario, the operator $\opE_m$ will behave as an identity in that subspace, while there is no strong enough constrain on the matrix elements of $\cvb{\opH_i,\opH_j}$ on this subspace, i.e. they can be nonzero (in degenerate subspace) but obey $\cvb{\opE_m,\cvb{\opH_i,\opH_j}}=0.$ Hence the degeneracy of $\opE_m$ is a condition for vanishing of $\cvb{\opH_i,\opH_j}.$
\end{proof}

From the proof above, when at least one of observables $\opE_m$ is degenerate, Eq.\eqref{eq:invariant_condition_nondegenerate} will no longer guarantee the simultaneous diagonalizability of conditioned Hamiltonians $\opH_i$ and $\opH_j.$ 
The nondegeneracy of the induced observables is sufficient to obtain the equivalence between KC and commutativity of $\mathcal{A}_{acc};$ however, it is not generally necessary---an example will be given in the following Section where we will consider the case of qubit being a probe.
We mention that the role of nondegeneracy also appears in the similar study \cite{Strasberg2019} on the relation between KC (as a definition of classicality) and incoherent dynamics of open quantum system (Eq.~\eqref{eq:necessary_Kol} in our case.)

We remark here that the nondegeneracy of $\opE_m$ does not require Hamiltonians $\opH_i$ to be nondegenerate. However, some particular configuration of energy spacings may lead to $\opE_m$  being degenerate. For example, by the definition of $\opE_m$ we have its eigenvalues
    \begin{equation}
        e_{\ell}(m) = \sum_{i,i'}\left \cvb{\overline{\gamma}(i)m(i)\overline{m}(i')\gamma(i')\right }e^{-it\cv{h_{\ell}\cv{i}-h_{\ell}\cv{i'}}} \label{eq:eigenvalues-E_m} \,\, ,
    \end{equation}
and one can see that for arbitrary meter states $\ket{m},$ there can be $\ell$ and $\ell'$ such that $e_{\ell}\cv{m}-e_{\ell'}\cv{m}=0$ if  \[h_{\ell}\cv{i}-h_{\ell'}\cv{i}=h_{\ell}\cv{i'}-h_{\ell'}\cv{i'}\] for all $i$ and $i',$ or the energy spacing between $\ell-$ and $\ell'-$levels is identical for all the conditional Hamiltonians. 

\subsection{Equivalence for the qubit case}\label{subsec:qubit_case}
In case of the probe being a qubit the accessible sub-algebra $\mathcal{A}_{acc}$ is generated by $\cvc{\iden,\opH_\uparrow,\opH_\downarrow},$ and we have
\[\op{E}^X_\pm=\dfrac{\iden}{2}\pm\dfrac{1}{4}\cv{\op{U}_\uparrow^\dagger\op{U}_\downarrow + \op{U}_\downarrow^\dagger\op{U}_\uparrow},~~\op{E}^Y_\pm=\dfrac{\iden}{2}\pm\dfrac{i}{4}\cv{\op{U}_\uparrow^\dagger\op{U}_\downarrow - \op{U}_\downarrow^\dagger\op{U}_\uparrow}\]
for $X-$ and $Y-$ measurements. From Lemma~\ref{props:single_axis}, KC for $\prob_n^{X}$ or $\prob_n^{Y}$ is equivalent $\cvb{\op{H}_{\uparrow(\downarrow)},\op{E}^X_\pm}=0$ or $\cvb{\op{H}_{\uparrow(\downarrow)},\op{E}^Y_\pm}=0,$ respectively. An interesting result is obtained when we assume that both of these conditions are fulfilled simultaneously.

\begin{proposition}\label{props:qubit_case}
    Let $\prob_n^\alpha\cv{m^\alpha_n,\ldots,m^\alpha_1}$ be a probability of measurement sequences for the observables $\opS_\alpha$ for $\alpha=X,Y$ of the qubit undergoing pure dephasing, with re-preparations in state $\ket{+x},$ i.e. $\ket{\gamma}=\ket{+x}$ for all $n.$ Both processes $\prob_n^X$ and $\prob_n^Y$ satisfy KC for arbitrary state $\rho$ if, and only if, $\mathcal{A}_{acc}$ is commutative.

\end{proposition}
\begin{proof}
If both the probabilities over sequences of repeating $X$ and $Y$ measurements $\prob_n^X$ and $\prob_n^Y$ satisfy KC condition, then from Lemma \ref{props:single_axis} we have 
	\begin{equation}
		\cvb{\opH_{\uparrow(\downarrow)},\op{E}^X_m\pm i\op{E}^Y_m}=0 \label{eq:sum_EXY}
	\end{equation}
or equivalently $\cvb{\opH_\downarrow,\op{U}_\uparrow}=\cvb{\opH_\downarrow,\op{U}^\dagger_\uparrow}=\cvb{\opH_\uparrow,\op{U}_\downarrow}=\cvb{\opH_\uparrow,\op{U}^\dagger_\downarrow}=0.$ Hence $\cvb{\opH_\uparrow,\opH_\downarrow}=0.$ 
\end{proof}

From this observation, in principle, one can characterize the quantumness/classicality of the accessible sub-algebra $\mathcal{A}_{acc}$ generated by $\cvc{\iden,\opH_\uparrow,\opH_\downarrow}$ from the measurement statistics of observables $\opS_x$ and $\opS_y$ for all possible $\rho.$ 
We stress that the proposition above does not require the nondegeneracy of the conditional Hamiltonian $\opH_{\uparrow(\downarrow)},$ but instead, KC needs to be verified for two types of measurements. 

It is interesting to extend this derivation into higher dimensional probe system. 
For instance, when the dimension $d$ is a power of prime numbers,  the operator $\opS_x$ can be generalized to a shift operator $\op{g}=\sum_m \ketbra{m+1 \mod d}{m}.$ One interesting possibility is to consider the measurements axes from the eigenstates of $\op{g},\op{h},\op{g}\op{h},\op{g}\op{h}^2,\ldots,\op{g}\op{h}^{d-1}$ where $\op{h}=diag\cv{1,\omega^1,\ldots,\omega^{d-1}}$ with $\omega=e^{2i\pi/d}.$
 These are  so-called mutually unbiased bases of measurements \cite{Bandyopadhyay2002,Bengtsson2007}. 
It is interesting whether there exists a combination similar to Eq.~\eqref{eq:sum_EXY} in the qudit case for such choice of measurement bases. If so, one can achieve a clear classification of the quantumness of a given quantum system by multi-dimensional probe (including a multi-qubit probe) without the assumption on nondegeneracy. However we leave this issue as an open question for future research.
\section{Witnessing noncommutativity of the accessible algebra of the probed system}\label{sec:QW}
A quantumness witness is a (operator-valued) measure on quantum system $\cv{\mathcal{A},\rho}$ characterizing whether the system is quantum or classical.
It was proposed by Alicki et al. in Refs.~\cite{Alicki2008b,Alicki2008a} in analogy to entanglement witness, and it was extensively studied by Facchi et al. in Refs.~\cite{Facchi2012,Facchi2013}. In general language, in the sense of definition of classicality used here and in Ref.~\cite{Alicki2008b}, a quantumness witness is an observable $\op{Q}\in\mathcal{A},$ for a given state $\rho,$ that has distinct values between commutative and noncommutative $\mathcal{A}$ i.e. $\tr\cv{\rho~\!\op{Q}}\in\mathbb{R}_{cl}$ for commutative $\mathcal{A}$ and $\tr\cv{\rho~\!\op{Q}}\in\mathbb{R}_{qu}$ for noncommutative $\mathcal{A}$ where $\mathbb{R}_{cl},\mathbb{R}_{qu}$ are some distinct subsets of the real line $\mathbb{R}_{cl},\mathbb{R}_{qu}\subset\mathbb{R}$ with $\mathbb{R}_{cl}\cap\mathbb{R}_{qu}=\emptyset.$
Ones of the examples are anti-commutations of arbitrary pair of positive operators, $\op{Q}=\cvc{\op{X},\op{Y}}$ for arbitrary $X>0$ and $\op{Y}>0,$ that all of them will be positive if and only if the underlying algebra $\mathcal{A}$ is commutative \cite[Theorem 1]{Alicki2008b}. Another example is associative test operator $\op{Q}=\opA\circ\cv{\op{B}\circ\op{C}}-\cv{\opA\circ\op{B}}\circ\op{C},$ where $\opA\circ\op{B}=\cvc{\opA,\op{B}}/2,$ that is vanishing with respect to $\rho$ for all triplets $\cv{\opA,\op{B},\op{C}}$ if and only if  $\mathcal{A}$ is commutative \cite{Facchi2013}. Note that both categories of quantumness witnesses cannot confirm classicality of the system unless the test is done for all pairs of positive operators in the first case and all triplets for the second category. 

From the results of the previous section it is clear that breaking of any of KC conditions, i.e.~a nonzero value of the below expression for any $n$ and $j$
\begin{eqnarray}
   \delta \prob_{n,j}& \equiv  \sum_{m_j} \prob_n\cv{m_n,\ldots,m_{j+1},m_j,m_{j-1},\ldots,m_1} - \nonumber\\  & \prob_{n-1}\cv{m_n,\ldots,m_{j+1},m_{j-1},\ldots,m_1} \,\, \label{KCwitness}
\end{eqnarray}
is a witness of noncommutativity of the algebra $\mathcal{A}_{acc}$ accessible by a probe (under additional assumption that  the observables $\opE_m$ are nondegenerate for all $m$ when the probe's Hilbert space has dimension larger than $2$). 
{The nonvanishing of quantity Eq.~\eqref{KCwitness} as a signature of the quantumness is employed in several works, e.g. the Fine’s theorem for Leggett-Garg tests \cite{Halliwell2019}, the witness of quantum coherence \cite{Li2012,Schild2015}, or recently as an experimental measure for the degree of quantumness \cite{Smirne2019}.} 
One can verify such equalities by comparing two types of set-ups: one with $n$ measurements and one with a subsequence of $n-1$ measurements with $j^{th}$ one from the first setup omitted. 

Of course, proving that $\mathcal{A}_{acc}$ is commutative from the analysis of KC of measurement on the probe is hard: it requires checking the consistency conditions for all possible states $\rho$ of the system. Let us however remind that as the proof of Lemma 1 implies that such checking has to be done only for the simplest form of KC condition that involves only $\prob_{2}\cv{m_2,m_1}$ and $\prob_{1}\cv{m_1}$.

Fortunately, the falsification of commutativity of the accessible algebra  of the system is more feasible in practice. Nonzero value of $\delta \prob_n$ for any $n$ and $\rho$ is a witness of noncommutativity, and one can hope that such a feature of $\mathcal{A}_{acc}$ becomes apparent for most states $\rho$ and most values of $n$. It is however easy to see that consistency conditions corresponding to various $n$ will differ in their sensitivity to noncommutativity of $\mathcal{A}_{acc}$. Let us give one example here: when $[\rho,\hat{K}_m] = 0,$ or equivalently $[\rho,H_i]=0,$
we have $\prob_2\cv{m_2,m_1}=\tr\cv{\rho\op{K}^\dagger_{m_1}\op{E}_{m_2}\op{K}_{m_1}}=\tr\cv{\rho\!~\op{F}_{m_1}\op{E}_{m_2}}$ where $\sum_{m}\op{F}_{m}=\iden$ because the measurement map considered here is unital. 
This means that KC for $n=2$ is always fulfilled for such $\rho$ (with a completely mixed state being an example), independent of properties of $\mathcal{A}_{acc}$. On the other hand, the same reasoning applied to $\prob_3$ shows that while KC for $n=3$ and $j=1$ is fulfilled for both commutative and noncommutative $\mathcal{A}_{acc}$, KC with $j=2$ is not automatically true, 
and its falsification will signify the noncommutativity of $\mathcal{A}_{acc}$.

\section{Discussion} \label{sec:discussion}
In this section we will discuss relations between the commutativity of the system's algebra accessible by the probe, and other possible definitions of classicality of the system as witnessed by probe (or probes) coupled to it.

\subsection{Relation to generation of probe-system entanglement during the evolution}
When a probe $P,$ initialized in a superposition of its pointer states $\ket{i}$ from Eq.~\eqref{eq:pure-dephase-U}, interacts with the system $S$ for a finite time, it experiences dephasing: its reduced density matrix $\rho_P$ becomes mixed, with its off-diagonal elements being suppressed relative to their initial values. When the initial state of the system $\rho$ is pure, this decoherence is necessarily accompanied by creation of $P$-$S$ entanglement \cite{Kuebler_AP73,Zurek_RMP03,Schlosshauer_book}. However, for a mixed state $\rho,$ dephasing can occur without such entanglement \cite{Eisert_PRL02,Pernice_PRA11,Roszak_PRA15,Roszak_qudit_PRA18,Roszak_PRA18}. Necessary and sufficient conditions on $\hat{U}_i$ and $\rho$ leading to zero $P$-$S$ entanglement during dephasing of $P$  are given in Refs.~\cite{Roszak_PRA15,Roszak_qudit_PRA18,Roszak_PRA18}. For $P$  being a qubit, there is no $P$-$S$ entanglement at time $t$ if, and only if, $\hat{U}_0(t) \rho \hat{U}_0^\dagger (t) = \hat{U}_1(t) \rho \hat{U}^\dagger_1(t)$, where $0$ and $1$ denote the two states of the qubit. 
For a higher-dimensional probe, apart from $\hat{U}_i \rho \hat{U}_i^\dagger = \hat{U}_j \rho \hat{U}^\dagger_j$ for any $i$ and $j$, we also need $[\hat{U}_i,\hat{U}_j]=0$ for any $i$ and $j$, i.e.~the accessible algebra has to be commutative. 

We arrive now at an interesting observation. For a qubit probe, commutativity of $\mathcal{A}_{acc}$ (and Kolmogorov consistency of measurements on $P$ ) is unrelated to $P$-$S$ entanglement being created (or not) during the probe's dephasing. The algebra can be commutative, but  when $[\rho,\hat{H}_0]\neq [\rho,\hat{H}_1]$, there will be $P$-$S$ entanglement. On the other hand, when $\rho$ is completely mixed, there is no such entanglement, no matter what the properties of the algebra are.
In contrast, for a higher-dimensional probe, commutativity of $\mathcal{A}_{acc}$ is necessary for separability of the joint state of $P$  and $S$ during pure dephasing of the probe. Breaking of KC for such a probe implies thus that $P$-$S$ entanglement is expected to be nonzero during the probe's dephasing. Let us note that the correlation function  $\Delta^\alpha_{2,1}$ was used to detect the entangling character of qubit-system evolution in Ref.~ \cite{Rzepkowski2020}. 
The discussion above shows that checking of KC of measurements on a higher-dimensional probe could be a more sensitive witness of entangling character of probe's dephasing. 
 
\subsection{The system being a source of classical noise acting on the probe}
Dephasing of quantum probes interacting with a quantum system (an environment) can often be modelled by replacing $S$ by a source of classical noise, the statistical properties of which are derived from $\rho$, $\hat{H}_S$, and $\hat{V}_S$ \cite{Szankowski_JPCM17}. For such a replacement to be robust, for example to be valid even when the probe is subjected to an external control (such as dynamical decoupling \cite{Viola_PRA98,Szankowski_JPCM17,Degen_RMP17} by application of short pulses rotating the qubit state between periods of its evolution due to coupling to $S$), certain criteria have to be fulfilled. For $S$ being of finite dimension considered here, it has been proven \cite{Szankowski_SR20} that $[\hat{H}_S,\hat{V}_S]\! =\!0,$ which is equivalent to commutativity of $\mathcal{A}_{acc}$, guarantees the existence of such noise representation. In fact, the commutativity condition corresponds to $S$ being replaced by a static random field when considering the dephasing of $P$. This is obvious from the fact that after letting the $P$ interact with $S$ for time $t$, the off-diagonal elements of the reduced density matrix of $P$, $\rho_P$, become
\begin{equation}
    \frac{\bra{i} \rho_{P}(t) \ket{j}}{ \bra{i} \rho_{P}(0) \ket{j}} =  e^{-i(\epsilon_i - \epsilon_j)t} \sum_n \bra{n}\rho_S(0)\ket{n} e^{-i(v_i-v_j)\nu_n t} \,\, ,
\end{equation}
where $\ket{n}$ are common eigenstates of $\hat{H}_S$ and $\hat{V}_S$ and $\hat{V}_S\vert n\rangle = \nu_n\vert n\rangle.$ The dephasing is thus of random unitary form, with phase shift by $(v_i-v_j)\nu_n t$ applied to $i$-$j$ coherence with probability $p_n \! \equiv \! \langle n\vert \rho_S(0)\vert n\rangle.$

\subsection{Leggett-Garg inequality}\label{subsec:LG}
Here we discuss the connection to one of standard methods of verification of quantumness in sequential measurement setups, namely the Leggett-Garg inequality \cite{Leggett1985,Emary2014}. 
{As previously discussed, the use of KC to witness the quantumness is known in literature \cite{Smirne2019,Halliwell2019,Li2012,Schild2015}. Here we demonstrate the similar concept for our probe-system mechanism with the example given in Ref.~\cite{Bechtold2016}.} 
Prior to doing so let us introduce another type of witness derived from KC. 
One can also consider as witnesses various combinations of correlation functions of multiple measurements on the probe. One example is 
    \begin{align}
		\Delta_{n,j}&=\sum_{m_k}\Big(\prod_{k\neq j}m_k\Big) \delta \prob_{n,j}
	\end{align}
	which is a difference of two correlation functions: one obtained by performing $n$ measurements, but discarding the result of $j$-th one, and evaluating the correlation of all the other results, and the other obtained by doing $n-1$ measurements at the same times (with $j$-th one not performed) and evaluating the correlation of all the results. Obviously, $\Delta_{n,j} \neq 0$ for some $n$ and $j \! \leq \! n$ implies breaking of KC and thus the noncommutativity of $\mathcal{A}_{acc}$.

For the qubit case with induced measurements, we can have two possible types of quantities $\Delta^X_{n,j}$ and $\Delta^Y_{n,j}$ from the different axes of measurements. From Proposition 1 we know that either $\Delta^X_{n,j}$ or $\Delta^Y_{n,j}$ is nonvanishing for some $n$ and $j\leq n$ if and only if $\mathcal{A}_{acc}$ is noncommutative.

For witnessing the noncommutativity, given that $\alpha=X,Y$, with two measurements  we can construct
	\begin{equation}
		\Delta_{2,1}^\alpha = \cvr{\sigma_\alpha(t_2)}_{t_1}-\cvr{\sigma_\alpha(t_2-t_1)} \label{eq:Delta_2}
	\end{equation}
where $\cvr{\sigma_\alpha(t_2)}_{t_1}=\sum_{m_2,m_1}m_2\prob^\alpha_2\cv{m_2,m_1}$ is an average of observable $\opS_\alpha$  at time $t_2$ after preparation, given that a nonselective measurement (followed by re-preparation) is done at time $t_1$ after the initial preparation, while $\cvr{\sigma_\alpha(t_2-t_1)}=\sum_{m_2}m_2\prob^\alpha_1\cv{m_2}$ is an average of observable $\opS_x$ with respect to the same initial state at time $t_2-t_1$ after preparation. However, since this witness relies on $\delta \prob_2 \neq 0$, it will not detect noncommutativity if $[\rho,\hat{H}_{\uparrow(\downarrow)}]=0.$
This can be avoided by considering a sequence of three measurements, and evaluating the quantities of the form		
    \begin{equation}
		\begin{split}
			\Delta_{3,2}^\alpha &= \cvr{\sigma_\alpha(t_3)\sigma_\alpha(t_1)}_{t_2}-\cvr{\sigma_\alpha(t_3-t_2)\sigma_\alpha(t_1)}\label{eq:Delta_3}\\
			&=\sum_{m_3,m_1}m_3m_1 \delta \prob^{\alpha}_{3,2}
		\end{split}
	\end{equation}
where \[\cvr{\sigma_\alpha(t_3)\sigma_\alpha(t_1)}_{t_2}=\sum_{m_3,m_2,m_1}m_3m_1\prob^\alpha_3\cv{m_3,m_2,m_1}\] is a correlation of measurements of observable $\opS_\alpha$ between first and third steps given a nonselective measurement (followed by re-preparation) at time $t_2;$ and  \[\cvr{\sigma_\alpha(t_3-t_2)\sigma_\alpha(t_1)}=\sum_{m_3,m_1}m_3m_1\prob^\alpha_2\cv{m_3,m_1}\] is a correlation of measurements with respect to the same initial state between time $t_1$ and $t_3-t_2.$ 
We demonstrate in the following that the Leggett-Garg inequality can be derived from the this type of witness.

Kolmogorov consistency of probabilities of obtaining various results of sequential projective measurements on the probe, can also be interpreted as a property of probability of sequential weak measurements on the system coupled to the probe, if one considers the induced map $\mathcal{K}_m$ a weak measurement on the system. 
From this perspective, since KC is an important ingredient in the derivation of the Leggett-Garg-type (LG-type) inequalities \cite{Leggett1985,Milz2019}, we can expect the connection between the noncommutativity witnesses $\Delta^\alpha_{2,1}$ or $\Delta^\alpha_{3,2}$ and the LG-type inequalities. 
This is the case of the experiment in Ref.~\cite{Bechtold2016} where the quantifier $\Delta^X_{2,1}$ can be computed from the measurement values assigned to be $m=0,1$ for the outcomes $\ket{-x}$ and $\ket{+x},$ denoted by $+$ and $-,$ respectively, and assuming post-selection without re-preparation at the second step. In particular,
	\begin{equation}
		\Delta_{2,1}^X=\prob_2\cv{+,+}-\prob_1\cv{+}+\prob_2\cv{-,-}.
	\end{equation}
Let us note that using notation from Ref.~\cite{Bechtold2016}, $\prob_2\cv{+,+}=g_3\cv{t_2,t_1}$ and $\prob_1\cv{+}=g_2\cv{t_2+t_1}$ where $t_2$ is a time duration between the second and first measurements, $t_1$ is the evolution time before the first measurement for the construction of $g_3;$ and $t_2+t_1$ is the evolution time before the first measurement in $g_2.$ 
In this sense, the vanishing of $\Delta_2^X=0,$ together with $\prob_2\cv{-,-}\geq 0,$ is then  identical to LG-type inequality in Ref.~\cite{Bechtold2016} as
\[\prob_2\cv{+,+}\leq\prob_1\cv{+}.\] 
This exemplifes that our type of quantumness witness can be the construction elements of LG-type inequality, or equivalently the testing for LG-type inequality is simply a testing for KC of the underlying process as suggested in Ref.~\cite{Milz2019}, and by our equivalence it is also (partially) a test for commutativity of accessible sub-algebra $\mathcal{A}_{acc}.$

In order to obtain a LG-type inequality  two conditions are required: (MR) macro-realism or the system (in particular the system $S$ in our case) having at least two predetermined distinguishable states, and (NIM) the measurement being noninvasive, i.e. the states before and after measurement has to be identical. Hence violation of LG-type inequality will falsify either (MR) or (NIM), and to further falsify MR one needs to clarify that the measurement is noninvasive. 
Another interpretation is that (MR) and (NIM) are inclusive, i.e. both conditions cannot be false together. 
For the detail on this issue we encourage interested readers to consult Ref.~\cite{Emary2014} and references therein. In our case, from microscopic point of view, it is clear that our induced map $\mathcal{K}_{m}$ is invasive (there can be a back-action in quantum case), so our argument on LG-type inequality is valid only the case when the conditions (MR) and (NIM) are inclusive. The interrogation in Ref.~\citep{Bechtold2016} is noninvasive, as it is of an ideal negative measurement type, i.e.~it leaves the state unchanged for at least one measurement outcome \cite{Emary2014,Leggett1985}, and violation of LG-type inequality will lead to falsification of macro-realism without assumption on inclusion.

However, at the level of measurement statistics, it can be claimed that KC per se is a statistical version of noninvasiveness condition, i.e. the intermediate measurement does not change the statistics of the later measurements \cite{Kofler2013}.
In particular, for the probability defined in Eq.~\eqref{eq:prob_def}, KC can be interpreted as so-called no-signaling-in-time condition in the level of quantum operation, the condition which the probability of a measurement is independent of the preceding measurement. 
Within this concept the violation of KC or its corresponding LG-type inequality will simply imply a violation of macro-realism. 
In addition, it has been proven recently that KC or no-signalling-in-time is equivalent to nondisturbance property of the intermediate measurement, defined when the observable associated with the latter measurement observable is invariant under nonselective measurements at all possible intermediate time steps \cite{Uola2019}.

\subsection{Objectivity of classicality of system}\label{subsec:objectivity}
It should be clear that the commutativity of $\mathcal{A}_{acc}$ does not guarantee that the whole algebra $\mathcal{A}$ of the system is commutative, however, but if $\mathcal{A}_{acc}$ is noncommutative, so is $\mathcal{A}.$ 
This observation leads to  a problem concerning objectivity of the classicality (understood here as commutativity of its algebra) of the system. 
If we assume that $\mathcal{A}$ is noncommutative, and there are two experimenters, agent $a$ and agent $b,$ accessing such a system through different probes, leading to two distinct sub-algebras $\mathcal{A}^a_{acc}$ and $\mathcal{A}^b_{acc},$ they maybe arrive at different conclusions about $\mathcal{A}$. For example, agent $a$ may find that $\mathcal{A}^a_{acc}$ is commutative (in principle, assuming that KC can be tested for all the system states $\rho$), but the agent $b$ may find it noncommutative. In fact, the definition the commutativity of $\mathcal{A}$ is objective but if $\mathcal{A}_{acc} \neq \mathcal{A}$ the conclusion about classicality from the experiment data, in general, cannot be applicable beyond $\mathcal{A}_{acc}$ and the classicality of $\mathcal{A}$ cannot be detectable even if KC for all $\rho$ is established. 
 When one agent makes a positive statement on commutativity, it should be understood as implying that the influence of the system on this agent's probe is ``classical'' in the sense used in this paper. For the quantumness, on the contrary, if $\mathcal{A}_{acc}$ fails to be commutative, the whole algebra $\mathcal{A}$ is noncommutative. 

The classicality (in the inter-subjective sense) of the system can be inferred from the consensus of observations on classicality from all possible observers equipped with different $\mathcal{A}_{acc},$ e.g. $\mathcal{A}$ is classical if all possible $\mathcal{A}_{acc}$ are classical. 
Let us note the fact of existence of a classical sub-algebra $\mathcal{A}_{min}\subset\mathcal{A},$ called minimal sub-algebra \cite{Alicki2008b}, which characterizes the classicality of the whole algebra $\mathcal{A}.$ If one could access this sub-algebra, the commutatvitiy tests thereon would also prove the classicality of the whole system $\mathcal{A}.$ 
However to confirm whether the accessible sub-algebra is minimal is an open problem, since there are still no established criteria to distinguish the minimal sub-algebra from the others \cite{Alicki2008b}. 

As an concrete example of a situation in which two agents interrogate the same system, let us consider a model inspired by physics of nitrogen-vacancy (NV) centre spin qubits in diamond \cite{Dobrovitski_ARCMP13,Rondin_RPP14} coupled to a system of nuclear spins. Under approximation that these spins are noninteracting---which holds up to a timescale given by inverse of typical dipolar coupling, and many experiments are done in this regime---the full algebra of the nuclear system is  generated by all the $I_{j}^k$ spin$-1/2$ operators of $^{13}$C nuclei, with $j = x,y,z$ and $k=1\ldots N$ labelling the nuclei. The Hamiltonian of the qubit probe (the NV enter) and the nuclear system is
\begin{equation}
\hat{H} = \Omega \hat{S}_z\otimes \mathds{I} + \mathds{1}\otimes \omega \sum_{k=1}^{N}\hat{I}_{z}^k + \sum_{k=1}^{N} \sum_{j=x,y,z} \hat{S}_z \otimes A^{j}_{k} \hat{I}^{k}_{j} \,\, , 
\end{equation}
where $\Omega$ is the qubit splitting, $\omega$ is the Zeeman splitting of the nuclei that is proportional to external magnetic field, and $A^{j}_{k}$ are the components of the hyperfine coupling between $k$-th nucleus and the qubit. 

After taking into account that the NV centre spin qubit is typically based on energy levels of spin-1 electronic complex with $S_z = 0$ and $1$, we have the conditional Hamiltonians $\hat{H}_0 = \omega \sum_k \hat{I}^{k}_z$ and $\hat{H}_{1} = \sum_k( \omega \hat{I}^{k}_z +\sum_{j}A^{j}_{k} \hat{I}^{j}_j)$, and in general these two do not commute. One situation in which the accessible algebra is commutative is when $\omega=0$, i.e.~at zero magnetic field. The other case is that of very large magnetic field, for which $\omega \! \gg \! A^{x,y}_k$, and these ``transverse'' hyperfine couplings have vanishing influence on the dynamics of the whole system. After approximating the Hamiltonian by its form with these couplings absent, we have $[\hat{H}_0,\hat{H}_1]\! =\! 0$. Two agents, one working at zero magnetic field, while the other is using a very large field, will then agree that the nuclear system, interrogated by making preparations and measurements on the probe qubit, is commutative, and the statistics of measurements on the qubit is classical---as we have shown by direct calculation of $\prob_n$ in \cite{Sakuldee2019}.
Agents probing the system at finite but not very large magnetic fields will disagree in general, as the above-described nuclear system at finite magnetic field is not classical.

\section{Conclusion and outlook}\label{sec:conclusion}
Let us summarize the main results of the paper. We have considered a quantum probe $P$ that  undergoes pure dephasing due to its interaction with a finite-dimensional quantum system $S$. Dynamics of $P$ is then described by a well-defined sub-algebra of operators of $S,$ i.e. the ``accessible'' algebra on $S$ from the point of view of $P.$ 
We have considered sequences of $n$ measurements on $P$ that induce POVMs described by effects $\hat{E}_m$ on $S,$ and investigated the relationship between Kolmogorov consistency of probabilities of obtaining these sequences that is assumed to hold for {\it an arbitrary} initial state $\rho$ of S, and commutativity of the accessible algebra. We have shown that when $\hat{E}_m$ are all nondegenerate, the Kolmogorov consistency and commutativity are equivalent. Furthermore, when $P$ is a qubit, Kolmogorov consistency for two sets of measurement sequences, corresponding to projections on eigenstates of $\hat{\sigma}_x$ and $\hat{\sigma}_y$ operators of the qubit, is equivalent to the commutativity of the accessible algebra. Observation of breaking of any of the consistency conditions in these cases can thus serve as a witness of noncommutativity of that algebra. 

Let us remark on several possible generalization from our results. One natural way is to consider infinite-dimensional systems. In fact, for countably infinite systems (i.e.~ infinite dimension with discrete structure), our observations can still be applied. It is a challenge to extend the description to system with continuous variables. Another interesting and highly nontrivial case is when the arbitrariness of the initial state $\rho$ is relaxed. In this scenario the problem becomes more complicated from both fundamental and practical points of view. 
For example, the Kolmogorov consistency cannot be transcribed onto a property of the measurement map, and the weaker definition of classicality connected with the initial state should be considered. 
If we also let the measurement maps be different at different times, it will lead to a nontrivial interplay between the initial state together with the pre-measurement history, the measurement maps, and the observables concerning future outcomes. We plan to investigate all these issues in future work, where we will study the connection between Kolmogorov consistency and several types of the commutativity of quantum systems.

\backmatter

\bmhead{Acknowledgments}
We would like to thank Piotr Sza{\'n}kowski for critical reading and valuable suggestions. Fruitful discussions with Simon Milz, Philip Taranto, and Dariusz Chru{\'s}ci{\'n}ski are appreciated. We acknowledge support by the Foundation for Polish Science (IRAP project, ICTQT, contract no. 2018/MAB/5, co-financed by EU within Smart Growth Operational Programme). Initial work on this topic was supported by funds from Polish National
Science Center, Grant No. 2015/19/B/ST3/03152.

\end{document}